\documentclass[copyright,creativecommons]{eptcs}
\providecommand{\event}{VPT 2020} 

\usepackage{stmaryrd}
\usepackage{latexsym}
\usepackage{amsmath}
\usepackage{amssymb}
\usepackage{wasysym}
\usepackage{proof}
\usepackage{tikz}

\newcommand{\expr}[1]{#1{}}

\newcommand{\var}[2]{\mathit{#1}#2}
\newcommand{\abs}[3]{\lambda #1{.#2{#3}}}
\newcommand{\app}[3]{#1{~#2{#3}}}

\newcommand{\args}[3]{#1{\ldots#2{#3}}}
\newcommand{\fundef}[2]{#1 & = & \expr{#2}}

\newcommand{\cas}[6]{\begin{array}[t]{@{\hspace*{0mm}}l@{\hspace*{1mm}}l@{\hspace*{1mm}}c@{\hspace*{1mm}}l@{\hspace*{0mm}}} 
\multicolumn{4}{@{\hspace*{0mm}}l@{\hspace*{0mm}}}{{\bf case} ~ #1 ~ {\bf of}} \\
& #2{} & \rightarrow & #3{} \\
~~~ | & #4{} & \rightarrow & #5{#6}\end{array}}
\newcommand{\longcas}[7]{\begin{array}[t]{@{\hspace*{0mm}}l@{\hspace*{1mm}}l@{\hspace*{1mm}}c@{\hspace*{1mm}}l@{\hspace*{0mm}}} 
\multicolumn{4}{@{\hspace*{0mm}}l@{\hspace*{0mm}}}{{\bf case} ~ #1 ~ {\bf of}} \\
& TT & \rightarrow & #2{} \\
~~~ | &  FF & \rightarrow & #3{} \\
~~~ | &  And(e,e') & \rightarrow & #4{} \\
~~~ | &  Or(e,e') & \rightarrow & #5{} \\
~~~ | & Ref(v) & \rightarrow & #6{#7}\end{array}}

\newcommand{\casedots}[6]{{\bf case}~#1{~{\bf of}~#2{\Rightarrow #3{~|\cdots|~#4{\Rightarrow #5{#6}}}}}}
\newcommand{\where}[3]{\!\!\!\begin{array}[t]{@{\hspace*{0mm}}l@{\hspace*{1mm}}c@{\hspace*{1mm}}l@{\hspace*{0mm}}}\multicolumn{3}{@{\hspace*{0mm}}l@{\hspace*{0mm}}}{#1{}}\\\multicolumn{3}{@{\hspace*{0mm}}l@{\hspace*{0mm}}}{{\bf where}}\\#2{#3}\end{array}}

\newcommand{\wheredots}[6]{#1{} ~ {\bf where} ~ #2 = #3, \ldots, #4 = #5{#6}}
\newcommand{\Letexp}[4]{{\bf let} ~ #1{ =  #2{~ {\bf in} ~ #3{#4}}}}

\newcommand{\brackets}[2]{(#1{)#2}}
\newcommand{\ltsprog}[2]{{\mathcal L}_p[\![#1{]\!]#2}}
\newcommand{\ltsexp}[4]{{\mathcal L}_e[\![#1{]\!]~#2{~#3{#4}}}}

\newcommand{\ignore}[1]{}

\newenvironment{proof}{{\em Proof}.~}{\hfill $\Box$\\}

\newtheorem{theorem}{Theorem}[section]

\newtheorem{lemma}[theorem]{Lemma}
\newtheorem{definition}[theorem]{Definition}  
\newtheorem{example}{Example}

\title{Distilling Programs to Prove Termination}

\author{G.W. Hamilton
\institute{School of Computing\\
Dublin City University\\
Ireland}
\email{hamilton@computing.dcu.ie}}

\def\titlerunning{Distilling Programs to Prove Termination}

\def\authorrunning{G.W. Hamilton}

\begin{document}

\maketitle

\begin{abstract}
The problem of determining whether or not any program terminates was shown to be undecidable by Turing, but recent advances in the area have allowed this 
information to be determined for a large class of programs. The classic method for deciding whether a program terminates dates back to Turing himself and 
involves finding a {\em ranking function} that maps a program state to a well-order, and then proving that the result of this function decreases for every possible program transition. 
More recent approaches to proving termination have involved moving away from the search for a single ranking function and toward a search for a set of ranking functions; this set 
is a choice of ranking functions and a disjunctive termination argument is used. In this paper, we describe a new technique for determining whether programs terminate. Our
technique is applied to the output of the {\em distillation} program transformation that converts programs into a simplified form called {\em distilled form}. Programs in distilled form are 
converted into a corresponding {\em labelled transition system} and termination can be demonstrated by showing that all possible infinite {\em traces} through this labelled transition 
system would result in an infinite descent of well-founded data values. We demonstrate our technique on a number of examples, and compare it to previous work.
\end{abstract}

\section{Introduction}

The {\em program termination problem}, or {\em halting problem}, can be defined as follows:
using only a finite amount of time, determine whether a given program will always finish running
or could execute forever. This problem rose to prominence before the development of stored program 
computers, in the time of Hilbert's {\em Entscheidungs problem}: the challenge to formalise all of mathematics 
and use algorithmic means to determine the validity of all statements. The halting problem was famously shown
to be undecidable by Turing \cite{TURING36}.

Although it is not possible to prove program termination in all cases, there are many programs for which this can be proved. 
The classic method for doing this dates back to Turing himself \cite{TURING48} and involves finding a ranking function that maps 
a program state to a well-order, and then proving that the result of this function decreases for every possible program transition.
This has a number of useful applications, such as in program verification, where {\em partial correctness} is often proved using 
deductive methods and a separate proof of termination is given to show {\em total correctness}, as originally done by Floyd \cite{FLOYD67}.
More recent approaches to proving termination have involved moving away from the search for a single ranking function and toward a search 
for a set of ranking functions; this set is a choice of ranking functions and a disjunctive termination argument is used. 
Program termination techniques have been developed for functional programs \cite{GEISL95,LEE01,GEISL06,MANOLIOS06},
logic programs \cite{SAGIV91,CODISH97,LINDENSTRAUSS97}, term rewriting systems \cite{DERSHOWITZ87,ARTS97,STEINBACH95} 
and imperative programs \cite{FLOYD67,BRADLEY05A,BERDINE06,COOK06,ALBERT08,SPOTO10,COOK11,GEISL14}. 

In this paper, we describe a new approach to the termination analysis of functional programs that is applied to the output of the {\em distillation} 
program transformation \cite{HAMILTON07A,HAMILTON12}. Distillation converts programs into a simplified form called {\em distilled form},
and to prove that programs in this form terminate, we convert them into a corresponding {\em labelled transition system} and then show 
that all possible infinite {\em traces} through the labelled transition system would result in an infinite descent of well-founded data values. 
This proof of termination is similar to that described in \cite{BROTHERSTON08} using {\em cyclic proof} techniques. However, we are
able to prove termination for a wider class of programs. \\
\\
The language used throughout this paper is a call-by-name higher-order functional language with the following syntax. 
\begin{definition}[Language Syntax]
\normalfont{The syntax of this language is as shown in Fig. \ref{grammar}.} \hfill $\Box$
\end{definition}
\begin{figure}[htb]
\begin{center}
\begin{tabular}{@{\hspace*{0mm}}l@{\hspace*{1mm}}r@{\hspace*{1mm}}l@{\hspace*{1mm}}l@{\hspace*{0mm}}}
$\expr{\var{prog}}$ & ::= & $\expr{\wheredots{\var{e_{0}}}{\var{h_1}}{\var{e_1}}{\var{h_k}}{\var{e_k}}}$ & Program \\
\\
$\expr{\var{e}}$ & ::= & $\expr{\var{x}}$ & Variable \\
& $|$ & $\expr{\app{\var{c}}{\args{\var{e_1}}{\var{e_n}}}}$ & Constructor Application \\
& $|$ & $\expr{\abs{\var{x}}{\var{e}}}$ & $\lambda$-Abstraction \\
& $|$ & $\expr{\var{f}}$ & Function Call \\
& $|$ & $\expr{\app{\var{e_0}}{\var{e_1}}}$ & Application \\
& $|$ & $\expr{\casedots{\var{e_0}}{\var{p_1}}{\var{e_1}}{\var{p_n}}{\var{e_n}}}$ & Case Expression \\ 
& $|$ & $\expr{\Letexp{\var{x}}{\var{e_0}}{\var{e_1}}}$ & Let Expression \\
\\
$\expr{\var{h}}$ & ::= & $\expr{\app{\var{f}}{\args{\var{x_1}}{\var{x_n}}}}$ & Function Header \\
\\
$\expr{\var{p}}$ & ::= & $\expr{\app{\var{c}}{\args{\var{x_1}}{\var{x_n}}}}$ & Pattern
\end{tabular}
\end{center}
\caption{Language Syntax}
\label{grammar}
\end{figure} 
Programs in the language consist of an expression to evaluate and a set of function definitions.
An expression can be a variable, constructor application, $\lambda$-abstraction, function call, application, {\bf case} or {\bf let}. 
Variables introduced by function headers, $\lambda$-abstractions, {\bf case} patterns and {\bf let}s are {\em bound}; all other variables are {\em free}. 
An expression that contains no free variables is said to be {\em closed}. We write $e \equiv e'$ if $e$ and $e'$ differ only in the names of bound variables.

Each constructor has a fixed arity; for example $\expr{\var{Nil}}$ has arity 0 and $\expr{\var{Cons}}$ has arity 2. 
In an expression $\expr{\app{\var{c}}{\args{\var{e_{1}}}{\var{e_{n}}}}}$, $n$ must equal the arity of $c$. 
The patterns in {\bf case} expressions may not be nested.  No variable may appear more than once within a pattern. 
We assume that the patterns in a {\bf case} expression are non-overlapping and exhaustive.
It is also assumed that erroneous terms such as $\expr{\app{\brackets{\app{\var{c}}{\args{\var{e_1}}{\var{e_n}}}}}{\var{e}}}$ 
where $c$ is of arity $n$ and $\expr{\casedots{\brackets{\abs{x}{e}}}{\var{p_1}}{\var{e_1}}{\var{p_n}}{\var{e_n}}}$ cannot occur.

\begin{example}
\normalfont{Consider the program from \cite {BRADLEY05B} shown in Figure \ref{example} for calculating the greatest common divisor of two numbers $x$ and $y$.
\begin{figure}[htb]
\begin{center}
\begin{tabular}{l}
$\expr{\where{\app{\app{\var{gcd}}{\var{x}}}{\var{y}}}
{\fundef{\app{\app{\var{gcd}}{\var{x}}}{\var{y}}}{\cas{\brackets{\app{\app{\var{gt}}{\var{x}}}{\var{y}}}}{\var{True}}{\app{\app{\var{gcd}}{\brackets{\app{\app{\var{sub}}{\var{x}}}{\var{y}}}}}{\var{y}}}{\var{False}}{\cas{\brackets{\app{\app{\var{gt}}{\var{y}}}{\var{x}}}}{\var{True}}{\app{\app{\var{gcd}}{\var{x}}}{\brackets{\app{\app{\var{sub}}{\var{y}}}{\var{x}}}}}{\var{False}}{\var{x}}}} \\
\fundef{\app{\app{\var{gt}}{\var{x}}}{\var{y}}}{\cas{\var{x}}{\var{Zero}}{\var{False}}{\app{\var{Succ}}{\var{x}}}{\cas{\var{y}}{\var{Zero}}{\var{True}}{\app{\var{Succ}}{\var{y}}}{\app{\app{\var{gt}}{\var{x}}}{\var{y}}}}} \\
\fundef{\app{\app{\var{sub}}{\var{x}}}{\var{y}}}{\cas{\var{y}}{\var{Zero}}{\var{x}}{\app{\var{Succ}}{\var{y}}}{\cas{\var{x}}{\var{Zero}}{\var{Zero}}{\app{\var{Succ}}{\var{x}}}{\app{\app{\var{sub}}{\var{x}}}{\var{y}}}}}
}}$
\end{tabular}
\end{center}
\caption{Example Program}
\label{example}
\end{figure}
Proving the termination of this program is tricky as there is no clear continued decrease in the size of either of the parameters of the {\em gcd} function (even though a number is subtracted from one of the arguments in each recursive call, it is difficult to determine that the number subtracted must be non-zero). 
We show how the termination of this program can be proved using our approach.}
\end{example}
The remainder of this paper is structured as follows. In Section 2, we give some preliminary definitions that are used throughout the paper.  In Section 3, we define the {\em labelled transition systems}  that are used in our termination proofs. In Section 4, we show how to prove termination of programs using our technique, and apply this technique to the program in Figure \ref{example}. In Section 5, we give some examples of programs that cause difficulties in termination analysis using other techniques, but are shown to terminate using our technique.
Section 6 concludes and considers related work.

\section{Preliminaries}

In this section, we complete the presentation of our programming language and give a brief overview of the distillation program transformation algorithm. 

\subsection{Language Definition}

\begin{definition}[Substitution]
\normalfont{We use the notation $\theta = \{x_1 \mapsto e_1, \ldots, x_n \mapsto e_n\}$ to denote a {\em substitution}.
If $e$ is an expression, then $e\theta = e\{x_1 \mapsto e_1, \ldots, x_n \mapsto e_n\}$ is the result of simultaneously 
substituting the expressions $e_1,\ldots, e_n$ for the corresponding variables $x_1,\ldots,x_n$, respectively, 
in the expression $e$ while ensuring that bound variables are renamed appropriately to avoid name capture.} \hfill $\Box$
\end{definition}
\begin{definition}[Language Semantics]
\normalfont{The call-by-name operational semantics of our language is standard: we define an evaluation relation $\Downarrow$ between 
closed expressions and {\em values}, where values are expressions in {\em weak head normal form} (i.e. constructor applications or $\lambda$-abstractions). 
We define a one-step reduction relation $\overset{r}{\leadsto}$ inductively as shown in Fig. \ref{reduction}, where the reduction $r$ can be $f$ (unfolding of function $f$), 
$c$ (elimination of constructor $c$) or $\beta$ ($\beta$-substitution).} \hfill $\Box$
\begin{figure}[htb]
\begin{center}
\begin{tabular}{c@{\hspace*{1cm}}c@{\hspace*{1cm}}c}
$((\lambda x.e_0)~e_1) \overset{\beta}{\leadsto} (e_0\{x \mapsto e_1\})$ & $\infer{f \overset{f}{\leadsto} \lambda x_1 \ldots x_n.e}{f~x_1 \ldots x_n=e}$ & $\infer{(e_0~e_1) \overset{r}{\leadsto} (e_0'~e_1)}{e_0 \overset{r}{\leadsto} e_0'}$  \\
\\
\multicolumn{3}{c}{$\infer{(\mathbf{case}~(c~e_1 \ldots e_n)~\mathbf{of}~p_1:e_1' | \ldots | p_k:e_k') \overset{c}{\leadsto} (e_i\{x_1 \mapsto e_1,\ldots,x_n \mapsto e_n\})}{p_i=c~x_1 \ldots x_n}$} \\
\\
\multicolumn{3}{c}{$\infer{(\mathbf{case}~e_0~\mathbf{of}~p_1:e_1 | \ldots p_k:e_k) \overset{r}{\leadsto} (\mathbf{case}~e_0'~\mathbf{of}~p_1:e_1 | \ldots p_k:e_k)}{e_0 \overset{r}{\leadsto} e_0'}$} \\
\\
\multicolumn{3}{c}{$(\mathbf{let}~x=e_0~\mathbf{in}~e_1) \overset{\beta}{\leadsto} (e_1\{x \mapsto e_0\})$}
\end{tabular} 
\end{center}
\caption{One-Step Reduction Relation}
\label{reduction}
\end{figure} 
\end{definition} 
\begin{definition}[Context]
\normalfont{A context $C$ is an expression with a ``hole" [] in the place of one sub-expression. $C[e]$ is the expression 
obtained by replacing the hole in context $C$ with the expression $e$. The free variables within $e$ may become bound within
$C[e]$; if $C[e]$ is closed then we call it a {\em closing context} for $e$.}
\end{definition}
We use the notation $e\overset{r}{\leadsto}$ if the expression $e$ reduces, $e\Uparrow$ if $e$ diverges, $e\Downarrow$ if $e$ converges and 
$e\Downarrow v$ if $e$ evaluates to the value $v$. These are defined as follows, where $\overset{r}{\leadsto}^*$ denotes the reflexive transitive closure of $\overset{r}{\leadsto}$:
\begin{center}
\begin{tabular}{l@{\hspace{2cm}}l}
$e \overset{r}{\leadsto}$, iff $\exists e'.e \overset{r}{\leadsto} e'$ & $e \Downarrow$, iff $\exists v.e \Downarrow v$ \\
$e \Downarrow v$, iff $e \overset{r}{\leadsto}^* v \wedge \neg(v\overset{r}{\leadsto})$ & $e \Uparrow$, iff $\forall e'.e \overset{r}{\leadsto}^* \!e' \Rightarrow e'\overset{r}{\ \leadsto\ }$
\end{tabular}
\end{center}
\begin{definition}[Contextual Equivalence]
\normalfont{Contextual equivalence, denoted by $\simeq$, equates two expressions if and only if they exhibit the same termination behaviour in all closing contexts i.e. 
$e_1 \simeq e_2$ iff $\forall C\ .\ C[e_1] \Downarrow$ iff $C[e_2] \Downarrow.$}
\end{definition}

\subsection{Distillation}

Distillation \cite{HAMILTON07A,HAMILTON12} is a powerful program transformation technique that builds on top of the positive supercompilation transformation algorithm \cite{TURCHIN86,SORENSEN96}. 
The following theorems have previously been proved about the distillation transformation ${\cal D}$.
\begin{theorem}[Correctness of Transformation]
$\forall p \in Prog: {\cal D}[\![p]\!] \simeq p$
\end{theorem}
Thus, the resulting program will have the same termination properties as the original program in all contexts.
\begin{theorem}[On The Form of Expressions Produced by Distillation]
\label{distilled}
\normalfont{For all possible input programs, distillation terminates and the form of expressions it produces (after all function arguments that are not variables are extracted using {\bf let}s), which we call {\em distilled form}, is described by $e^\emptyset$ where $e^{\rho}$ is defined as follows:
\begin{center}
\begin{tabular}{lrl}
$\expr{\var{e^{\rho}}}$ & ::= & $\expr{\app{\var{x}}{\args{\var{e_1^{\rho}}}{\var{e_n^{\rho}}}}}$ \\
& $|$ & $\expr{\app{\var{c}}{\args{\var{e_1^{\rho}}}{\var{e_n^{\rho}}}}}$ \\
& $|$ & $\expr{\abs{\var{x}}{\var{e^{\rho}}}}$ \\
& $|$ & $\expr{\app{\var{f}}{\args{\var{x_1}}{\var{x_n}}}}$ (where $f$ is defined by $f~x_1 \ldots x_n = e^{\rho}$) \\
& $|$ & $\expr{\casedots{\brackets{\app{\var{x}}{\args{\var{e_1^{\rho}}}{\var{e_n^{\rho}}}}}}{\var{p_1}}{\var{e_{{n+1}}^{\rho}}}{\var{p_k}}{\var{e_{n+k}^{\rho}}}}$ $(x \notin \rho)$ \\ 
& $|$ & $\expr{\Letexp{\var{x}}{\var{e_0^{\rho}}}{\var{e_1^{(\rho \cup \{x\})}}}}$
\end{tabular} 
\end{center}}
\end{theorem}
The particular property of expressions in distilled form that makes them easier to analyse for termination is that no sub-expression that has been extracted using a {\bf let} expression
can be an intermediate data structure; {\bf let} variables are added to the set $\rho$, and cannot be used in the selectors of {\bf case} expressions. 
This means that once a parameter has increased in size it cannot subsequently decrease, which makes it much easier to identify parameters that decrease in size.

Due to space considerations, we are not able to include a full definition of the distillation algorithm here. However, we can simply treat this as a black box that does not alter the 
termination properties of a program and will always convert it into distilled form, so this paper is still reasonably self-contained.

\begin{example}
\normalfont{The result of transforming the example program in Figure \ref{example} is shown in Figure \ref{exampletrans}. We can see that this program is indeed in distilled form.}
\end{example}
\begin{figure}[htb]
\begin{center}
\begin{tabular}{l}
$\expr{\where{\app{\app{\app{\app{\var{f0}}{\var{x}}}{\var{y}}}{\var{x}}}{\var{y}}}{
\fundef{\app{\app{\app{\app{\var{f0}}{\var{a}}}{\var{b}}}{\var{c}}}{\var{d}}}{\cas{\var{a}}{\var{Zero}}{\cas{\var{b}}{\var{Zero}}{\var{c}}{\app{\var{Succ}}{\var{b}}}{\app{\app{\app{\app{\var{f1}}{\var{c}}}{\var{b}}}{\var{c}}}{\var{b}}}}{\app{\var{Succ}}{\var{a}}}{\cas{\var{b}}{\var{Zero}}{\app{\app{\app{\app{\var{f3}}{\var{a}}}{\var{d}}}{\var{a}}}{\var{d}}}{\app{\var{Succ}}{\var{b}}}{\app{\app{\app{\app{\var{f0}}{\var{a}}}{\var{b}}}{\var{c}}}{\var{d}}}}} \\
\fundef{\app{\app{\app{\app{\var{f1}}{\var{a}}}{\var{b}}}{\var{c}}}{\var{d}}}{\cas{\var{a}}{\var{Zero}}{\app{\app{\app{\app{\var{f0}}{\var{c}}}{\var{b}}}{\var{c}}}{\var{b}}}{\app{\var{Succ}}{\var{a}}}{\app{\app{\app{\app{\var{f2}}{\var{a}}}{\var{b}}}{\var{c}}}{\var{d}}}} \\
\fundef{\app{\app{\app{\app{\var{f2}}{\var{a}}}{\var{b}}}{\var{c}}}{\var{d}}}{\cas{\var{a}}{\var{Zero}}{\cas{\var{b}}{\var{Zero}}{\var{c}}{\app{\var{Succ}}{\var{b}}}{\app{\app{\app{\app{\var{f1}}{\var{c}}}{\var{b}}}{\var{c}}}{\var{b}}}}{\app{\var{Succ}}{\var{a}}}{\cas{\var{b}}{\var{Zero}}{\app{\app{\app{\app{\var{f5}}{\var{a}}}{\var{d}}}{\var{a}}}{\var{d}}}{\app{\var{Succ}}{\var{b}}}{\app{\app{\app{\app{\var{f2}}{\var{a}}}{\var{b}}}{\var{c}}}{\var{d}}}}} \\
\fundef{\app{\app{\app{\app{\var{f3}}{\var{a}}}{\var{b}}}{\var{c}}}{\var{d}}}{\cas{\var{b}}{\var{Zero}}{\app{\app{\app{\app{\var{f3}}{\var{a}}}{\var{d}}}{\var{a}}}{\var{d}}}{\app{\var{Succ}}{\var{b}}}{\app{\app{\app{\app{\var{f4}}{\var{a}}}{\var{b}}}{\var{c}}}{\var{d}}}} \\
\fundef{\app{\app{\app{\app{\var{f4}}{\var{a}}}{\var{b}}}{\var{c}}}{\var{d}}}{\cas{\var{a}}{\var{Zero}}{\cas{\var{b}}{\var{Zero}}{\app{\var{Succ}}{\var{c}}}{\app{\var{Succ}}{\var{b}}}{\app{\app{\app{\app{\var{f5}}{\var{c}}}{\var{b}}}{\var{c}}}{\var{b}}}}{\app{\var{Succ}}{\var{a}}}{\cas{\var{b}}{\var{Zero}}{\app{\app{\app{\app{\var{f3}}{\var{a}}}{\var{d}}}{\var{a}}}{\var{d}}}{\app{\var{Succ}}{\var{b}}}{\app{\app{\app{\app{\var{f4}}{\var{a}}}{\var{b}}}{\var{c}}}{\var{d}}}}} \\
\fundef{\app{\app{\app{\app{\var{f5}}{\var{a}}}{\var{b}}}{\var{c}}}{\var{d}}}{\cas{\var{a}}{\var{Zero}}{\cas{\var{b}}{\var{Zero}}{\app{\var{Succ}}{\var{c}}}{\app{\var{Succ}}{\var{b}}}{\app{\app{\app{\app{\var{f5}}{\var{c}}}{\var{b}}}{\var{c}}}{\var{b}}}}{\app{\var{Succ}}{\var{a}}}{\cas{\var{b}}{\var{Zero}}{\app{\app{\app{\app{\var{f5}}{\var{a}}}{\var{d}}}{\var{a}}}{\var{d}}}{\app{\var{Succ}}{\var{b}}}{\app{\app{\app{\app{\var{f5}}{\var{a}}}{\var{b}}}{\var{c}}}{\var{d}}}}}}}$
\end{tabular}
\end{center}
\caption{Example Program Distilled}
\label{exampletrans}
\end{figure}

\section{Labelled Transition Systems}

In this section, we define the labelled transition systems used in our termination analysis.
\begin{definition}[Labelled Transition System]
\normalfont{A {\em labelled transition system} (LTS) is a 4-tuple \\ $({\cal E},e_0,Act,\to)$ where:}
\end{definition}
\begin{itemize}
\item ${\cal E}$ is a set of {\em states} of the LTS. Each is an expression or the end-of-action state {\bf 0}.
\item $e_0 \in {\cal E}$ is the {\em start state}.
\item $Act$ is a set of {\em actions} which can be one of the following:
\begin{itemize}
\item $x$, a variable; 
\item $c$, a constructor;
\item $\lambda x$, a $\lambda$-abstraction; 
\item $f$, a function unfolding;
\item $@$, the function in an application;
\item $\# i$, the $i^{th}$ argument in an application; 
\item $\mathbf{case}$, a case selector;
\item $p$, a case pattern;
\item $\mathbf{let}~x$, a let variable
\item $\mathbf{in}$, a let body.
\end{itemize}
\item $\to \ \subseteq {\cal E} \times Act \times {\cal E}$ is a {\em transition relation}. We write $e \xrightarrow{\alpha} e'$ for a 
transition from state $e$ to state $e'$ via action $\alpha$.  \hfill $\Box$
\end{itemize}
We also write $e \rightarrow (\alpha_1,t_1), \ldots, (\alpha_n,t_n)$ for a LTS with start state $e$ where $t_1 \ldots t_n$ 
are the LTSs obtained by following the transitions labelled $\alpha_1 \ldots \alpha_n$ respectively from $e$. 
\begin{definition}[Renaming]
\normalfont{We use the notation $\sigma = \{x_1 \mapsto x_1', \ldots, x_n \mapsto x_n'\}$ to denote a {\em renaming}.
If $e$ is an expression, then $e\sigma = e\{x_1 \mapsto x_1', \ldots, x_n \mapsto x_n'\}$ is the result of simultaneously replacing the 
free variables $x_1 \ldots x_n$ with the corresponding variables $x_1' \ldots x_n'$ respectively, in the expression $e$ while ensuring 
that bound variables are renamed appropriately to avoid name capture.} \hfill $\Box$
\end{definition}
\begin{definition}[Folded LTS]
\normalfont{A {\em folded LTS} is a LTS which also contains renamings of the form $e \overset{\sigma}{\longrightarrow} e'$, where $\sigma$ is a renaming s.t. $e \equiv e' \sigma$.} \hfill $\Box$
\end{definition}
We now show how to generate the LTS representation of a program. It is assumed here that all function arguments in the program are variables;
it is always possible to extract non-variable function arguments using {\bf let}s to ensure that this is the case.
\begin{definition}[Generating LTS From Program]
\label{ltstoprog}
\normalfont{A LTS can be generated from a program $p$ as $\expr{\ltsprog{\var{p}}}$ using the rules as shown in Fig. \ref{proglts}. The rules ${\cal L}_e$ 
generate a LTS from an expression where the parameter $\rho$ is the set of previously encountered function calls and the parameter $\Delta$ is the set of function definitions. 
If a renaming of a previously memoised function call is encountered, no further transitions are added to the constructed LTS. Thus, the constructed LTS will always be a finite 
representation of the program.} \hfill $\Box$
\end{definition}
\begin{figure}[htb]
\begin{center}
\begin{tabular}{@{\hspace*{0mm}}l@{\hspace*{1mm}}c@{\hspace*{1mm}}l@{\hspace*{0mm}}}
$\expr{\ltsprog{\wheredots{\var{e_0}}{\var{f_1}}{\var{e_1}}{\var{f_n}}{\var{e_n}}}}$ & = & $\expr{\ltsexp{\var{e_0}}{\rho}{(\Delta \cup \{f_1=e_1,\ldots,f_n=e_n\})}}$ \\
\\
$\expr{\ltsexp{e=\var{x}}{\rho}{\Delta}}$ & = & $e \rightarrow (x,\mathbf{0})$ \\
$\expr{\ltsexp{e=\app{\var{c}}{\args{\var{e_{1}}}{\var{e_{n}}}}}{\rho}{\Delta}}$ & = & $e \rightarrow (c,{\bf 0}),(\# 1,\expr{\ltsexp{\var{e_1}}{\rho}{\Delta}}),\ldots,(\# n,\expr{\ltsexp{\var{e_n}}{\rho}{\Delta}})$ \\
$\expr{\ltsexp{e=\abs{\var{x}}{\var{e}}}{\rho}{\Delta}}$ & = & $e \rightarrow (\lambda x,\expr{\ltsexp{\var{e}}{\rho}{\Delta}})$ \\
$\expr{\ltsexp{e=\app{\var{f}}{\args{\var{x_1}}{\var{x_n}}}}{\rho}{\Delta}}$ & = & $\left\{\begin{array}{ll}
e \xrightarrow{\sigma} e',  & $if $\exists e' \in \rho.e'\sigma \equiv e \\
e \rightarrow (f,\expr{\ltsexp{\var{e'}}{(\rho \cup \{e\})}{\Delta}}), & $otherwise $(f \equiv \lambda x_1 \ldots x_n.e' \in \Delta)
\end{array}\right.$ \\
$\expr{\ltsexp{e=\app{\var{e_{0}}}{\var{e_{1}}}}{\rho}{\Delta}}$ & = & $e \rightarrow (@,\expr{\ltsexp{\var{e_0}}{\rho}{\Delta}}),(\# 1,\expr{\ltsexp{\var{e_1}}{\rho}{\Delta}})$ \\
\multicolumn{3}{@{\hspace*{0mm}}l@{\hspace*{0mm}}}{$\expr{\ltsexp{e=\casedots{\var{e_0}}{\var{p_{1}}}{\var{e_1}}{\var{p_{n}}}{\var{e_n}}}{\rho}{\Delta}}$} \\
& = & $e \rightarrow ({\bf case},\expr{\ltsexp{\var{e_0}}{\rho}{\Delta}}),(p_1,\expr{\ltsexp{\var{e_1}}{\rho}{\Delta}}),\ldots,(p_n,\expr{\ltsexp{\var{e_n}}{\rho}{\Delta}})$ \\
$\expr{\ltsexp{e=\Letexp{\var{x}}{\var{e_0}}{\var{e_1}}}{\rho}{\Delta}}$ & = & $e \rightarrow ({\bf let}~x,\expr{\ltsexp{\var{e_0}}{\rho}{\Delta}}),({\bf in},\expr{\ltsexp{\var{e_1}}{\rho}{\Delta}})$ 
\end{tabular}
\end{center}
\caption{LTS Representation of a Program}
\label{proglts}
\end{figure}
\begin{example}
\normalfont{The LTS generated for the distilled {\em gcd} program in Figure \ref{exampletrans} is shown in Figure \ref{examplelts}.}
\begin{figure}[htbp]
$$\includegraphics[scale=1.0,trim = 25 50 50 200]{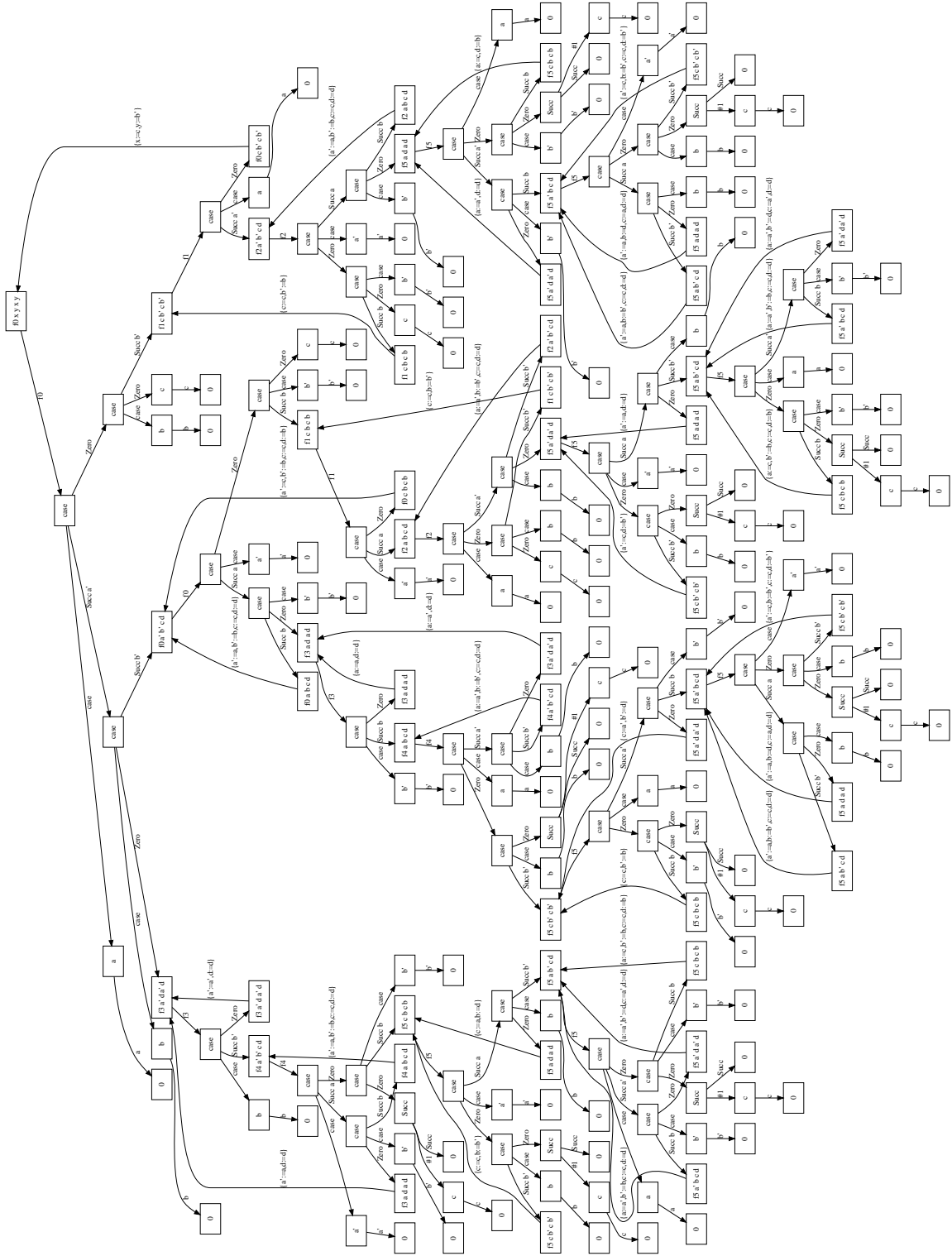}$$
\caption{LTS Representation of Example Program}
\label{examplelts}
\end{figure} 
\end{example}

\section{Proving Termination}

In order to prove that a program terminates, we analyse the labelled transition system generated from the result of transforming the program using distillation.
We need to show that within every cycle in this labelled transition system, at least one parameter is decreasing. We define a decreasing parameter as follows.
\begin{definition}[Decreasing Parameter]
\normalfont{A parameter is considered to decrease in size if it is the subject of a {\bf case} selector.} \hfill $\Box$
\end{definition}
A parameter that is the subject of a {\bf case} selector is deconstructed into smaller components and therefore decreases in size.
We define an increasing parameter as follows.
\begin{definition}[Increasing Parameter]
\normalfont{A parameter is considered to increase in size if any expression other than a variable is assigned to it.} \hfill $\Box$
\end{definition}
Note that this is a conservative criterion for an increase in size based on the syntactic size of the parameter rather than the semantic size. 
Thus, for example, in the call  {\em gcd (sub x y) y} , the first parameter would be considered to be increasing syntactically, even though it is actually decreasing semantically. 
However, such potentially increasing parameters are often transformed by distillation to reveal that they are in fact decreasing, as we have seen is the case for this example.
\begin{lemma}[On Decreasing Parameters]
\normalfont{Every parameter that has decreased in size cannot previously have increased in size.} \hfill $\Box$
\label{decreaselemma}
\end{lemma}
\begin{proof}
This can be proved quite straightforwardly from the definition of distilled form in Theorem \ref{distilled}. Within the distilled form $e^{\rho}$, if any expression other than
a variable is assigned to a parameter using a {\bf let}, then the parameter is added to the set $\rho$ and cannot subsequently be the subject of a {\bf case} selector. Thus, if a parameter
has increased in size, it cannot subsequently decrease.
\end{proof} \\
In order to prove that a program terminates, we need to show that all possible {\em traces} through the labelled transition system generated from the result of distilling the program are {\em infinitely progressing}.
We now define what these terms mean.
\begin{definition}[Trace]
\normalfont{A {\em trace} within a labelled transition system $({\cal E},e_0,Act,\to)$ is a sequence of states $e_0$, $e_1$, $\ldots$ where $\forall i.\exists \alpha.e_i \overset{\alpha}{\longrightarrow} e_{i+1} \in \to$.} \hfill $\Box$
\end{definition}
\begin{definition}[Infinitely Progressing Trace]
\normalfont{An {\em infinitely progressing trace} is a trace that contains an infinite number of decreases in parameter size.} \hfill $\Box$
\end{definition}
\begin{theorem}[Termination]
\normalfont{If all traces through the labelled transition system generated from the result of distilling a program are infinitely progressing, then the program terminates.} \hfill $\Box$
\end{theorem}
\begin{proof}
From Lemma \ref{decreaselemma}, every parameter that has decreased in size cannot previously have increased in size. If the trace is infinitely progressing,
then there must be an infinite number of decreases in parameter size. As these parameters cannot have increased in size elsewhere within the trace,
there must be infinite descent.
\end{proof} \\
Since a decreasing parameter must be the subject of a {\bf case} selector, to show that a program terminates it is sufficient to show that in the labelled transition system 
generated from the result of distilling the program there is a {\bf case} expression between every renamed state and its renaming.
\begin{example}
\normalfont{In the LTS generated from the distilled program in Figure \ref{exampletrans} is shown in Figure \ref{examplelts}, 
we can see that there is a {\bf case} expression between every renamed state and its renaming, so this program is indeed terminating.
Proving the termination of the original program is tricky as there is no clear continued decrease in the size of either of the parameters of the {\em gcd} function 
(even though a number is subtracted from one of the arguments in each recursive call, it is difficult to determine that the number subtracted must be non-zero).}
\end{example}

\section{Examples}

We now give some examples of programs that cause difficulties in termination analysis using other techniques, but can be shown to terminate using the technique described here.
None of these examples can be proven to terminate using the size-change principle described in \cite{LEE01}.
\ignore{
We concentrate on examples that cannot be proven to terminate using the size-change principle described in \cite{LEE01}, which is surprisingly effective and can even be used
to prove that the Ackermann function shown in Figure \ref{Ackermann} terminates despite the fact that it is not primitive recursive.
\begin{figure}[htb]
\begin{center}
$\expr{\where{\app{\app{\var{ack}}{\var{m}}}{\var{n}}}{\fundef{\app{\app{\var{ack}}{\var{m}}}{\var{n}}}{\cas{\var{m}}{\var{Zero}}{\app{\var{Succ}}{\var{n}}}{\app{\var{Succ}}{\var{m'}}}{\cas{\var{n}}{\var{Zero}}{\app{\app{\var{ack}}{\var{m'}}}{\brackets{\app{\var{Succ}}{\brackets{\var{Zero}}}}}}{\app{\var{Succ}}{\var{n'}}}{\app{\app{\var{ack}}{\var{m'}}}{\brackets{\app{\app{\var{ack}}{\var{m}}}{\var{n'}}}}}}} \\
}}$
\end{center}
\caption{Ackermann Function}
\label{ackermann}
\end{figure} 
}
\begin{example}
\label{example1}
\normalfont{Consider the following program:}
\begin{center}
$\expr{\where{\app{\var{f}}{\var{n}}}{\fundef{\app{\var{f}}{\var{n}}}{\cas{\var{n}}{\var{Zero}}{\var{Zero}}{\app{\var{Succ}}{\var{n'}}}{\app{\var{g}}{\brackets{\app{\var{Succ}}{\var{n}}}}}} \\
\fundef{\app{\var{g}}{\var{n}}}{\cas{\var{n}}{\var{Zero}}{\var{Zero}}{\app{\var{Succ}}{\var{n'}}}{\cas{\var{n'}}{\var{Zero}}{\var{Zero}}{\app{\var{Succ}}{\var{n''}}}{\app{\var{f}}{\var{n''}}}}} \\
}}$
\end{center}
This has mutually recursive functions $f$ and $g$, where the parameter is increasing in the call from $f$ to $g$, and decreasing in the call from $g$ to $f$
and therefore causes difficulties for other termination checkers.
The result of transforming this program using distillation is as follows:
\begin{center}
$\expr{\where{\app{\var{f}}{\var{n}}}{\fundef{\app{\var{f}}{\var{n}}}{\cas{\var{n}}{\var{Zero}}{\var{Zero}}{\app{\var{Succ}}{\var{n'}}}{\app{\var{f}}{\var{n'}}}} \\
}}$
\end{center}
The LTS generated for this transformed program is shown in Figure \ref{example1lts}. We can now quite easily see that this program terminates as there is a {\bf case} expression between the function call $f~n$ and its renaming $f~n'$.
\begin{figure}[htbp]
$$\includegraphics[scale=0.6,trim = -150 300 250 50]{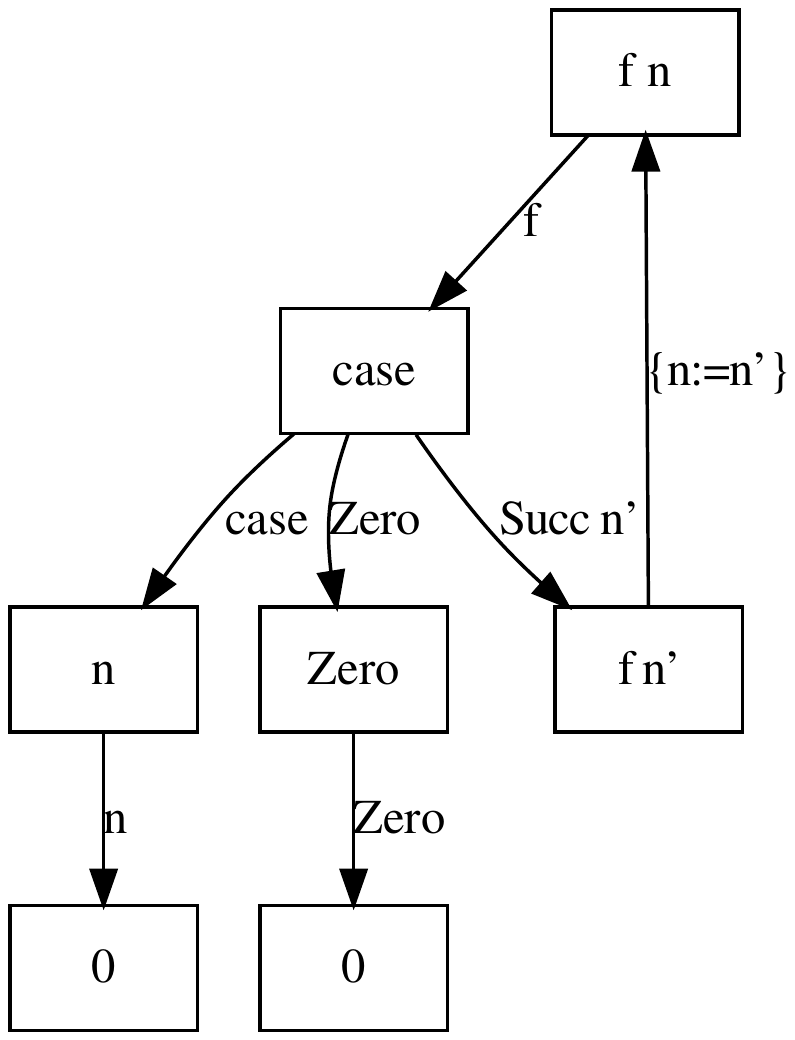}$$
\caption{LTS Representation of Program in Example \ref{example1}}
\label{example1lts}
\end{figure}
\end{example}
\begin{example}
\label{example2}
\normalfont{Consider the following program:}
\begin{center}
$\expr{\where{\app{\app{\var{f}}{\var{m}}}{\var{n}}}{\fundef{\app{\app{\var{f}}{\var{m}}}{\var{n}}}{\cas{\var{m}}{\var{Zero}}{\var{Zero}}{\app{\var{Succ}}{\var{m'}}}{\app{\app{\var{f}}{\brackets{\app{\app{\var{sub}}{\var{m}}}{\var{n}}}}}{\brackets{\app{\var{Succ}}{\var{n}}}}}} \\
\fundef{\app{\app{\var{sub}}{\var{x}}}{\var{y}}}{\cas{\var{y}}{\var{Zero}}{\var{x}}{\app{\var{Succ}}{\var{y}}}{\cas{\var{x}}{\var{Zero}}{\var{Zero}}{\app{\var{Succ}}{\var{x}}}{\app{\app{\var{sub}}{\var{x}}}{\var{y}}}}} \\
}}$
\end{center}
This causes problems for other termination checkers as the size of the second parameter is increasing and the size of the first parameter will not decrease if the value of the second parameter is {\em Zero}.
The result of transforming this program using distillation is as follows:
\begin{center}
$\expr{\where{\app{\app{\var{f}}{\var{m}}}{\var{n}}}{\fundef{\app{\app{\var{f}}{\var{m}}}{\var{n}}}{\cas{\var{m}}{\var{Zero}}{\var{Zero}}{\app{\var{Succ}}{\var{m'}}}{\cas{\var{n}}{\var{Zero}}{\app{\var{g}}{\var{m'}}}{\app{\var{Succ}}{\var{n'}}}{\app{\app{\var{f}}{\var{m'}}}{\var{n'}}}}} \\
\fundef{\app{\var{g}}{\var{m}}}{\cas{\var{m}}{\var{Zero}}{\var{Zero}}{\app{\var{Succ}}{\var{m'}}}{\app{\var{g}}{\var{m'}}}} \\
}}$
\end{center}
The LTS generated for this transformed program is shown in Figure \ref{example2lts}.
We can see that there is a case expression between every renamed state and its renaming, so this program is indeed terminating.
\begin{figure}[htbp]
$$\includegraphics[scale=0.6,trim = -150 150 250 50]{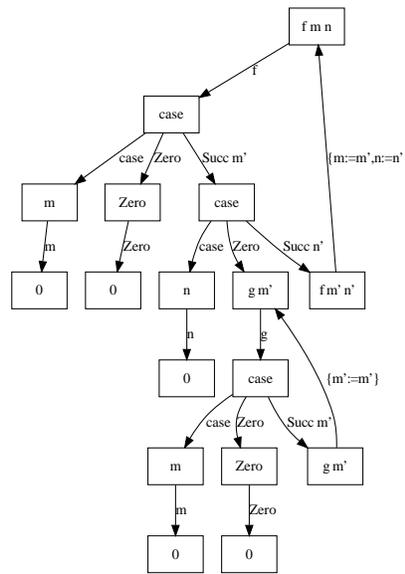}$$
\caption{LTS Representation of Program in Example \ref{example1}}
\label{example2lts}
\end{figure} 
\end{example}

\begin{example}
\label{example3}
\normalfont{Consider the following program:}
\begin{center}
$\expr{\where{\app{\app{\var{f}}{\var{m}}}{\var{n}}}{\fundef{\app{\app{\var{f}}{\var{m}}}{\var{n}}}{\cas{\var{m}}{\var{Zero}}{\var{Zero}}{\app{\var{Succ}}{\var{m'}}}{\cas{\var{n}}{\var{Zero}}{\app{\app{\var{f}}{\var{m'}}}{\var{n}}}{\app{\var{Succ}}{\var{n'}}}{\cas{\brackets{\app{\app{\var{gt}}{\var{m}}}{\var{n}}}}{\var{True}}{\app{\app{\var{f}}{\var{m'}}}{\var{n}}}{\var{False}}{\app{\app{\var{f}}{\brackets{\app{\var{Succ}}{\var{m}}}}}{\var{n'}}}}}} \\
\fundef{\app{\app{\var{gt}}{\var{x}}}{\var{y}}}{\cas{\var{x}}{\var{Zero}}{\var{False}}{\app{\var{Succ}}{\var{x}}}{\cas{\var{y}}{\var{Zero}}{\var{True}}{\app{\var{Succ}}{\var{y}}}{\app{\app{\var{gt}}{\var{x}}}{\var{y}}}}} \\
}}$
\end{center}
In the function $f$, the first parameter both increases and decreases, so this causes problems for other termination checkers.
The result of transforming this program using distillation is as follows:
\begin{center}
$\expr{\where{\app{\app{\var{f}}{\var{m}}}{\var{n}}}{\fundef{\app{\app{\var{f}}{\var{m}}}{\var{n}}}{\cas{\var{m}}{\var{Zero}}{\var{Zero}}{\app{\var{Succ}}{\var{m'}}}{\cas{\var{n}}{\var{Zero}}{\app{\var{g}}{\var{m'}}}{\app{\var{Succ}}{\var{n'}}}{\app{\app{\var{f}}{\var{m'}}}{\var{n'}}}}} \\
\fundef{\app{\var{g}}{\var{m}}}{\cas{\var{m}}{\var{Zero}}{\var{Zero}}{\app{\var{Succ}}{\var{m'}}}{\app{\var{g}}{\var{m'}}}} \\
}}$
\end{center}
The LTS generated for this transformed program is shown in Figure \ref{example3lts}.
We can see that there is a case expression between every renamed state and its renaming, so this program is indeed terminating.
\begin{figure}[htbp]
$$\includegraphics[scale=0.5,trim = -150 150 250 50]{LTS4.pdf}$$
\caption{LTS Representation of Program in Example \ref{example1}}
\label{example3lts}
\end{figure} 
\end{example}

\begin{example}
\label{example4}
\normalfont{The final example shown in Figure \ref{mccarthy} is McCarthy's 91 function, which is nested recursive and has often been used as a test case for proving termination.}
\begin{figure}[htb]
\begin{center}
$\expr{\where{\app{\var{f}}{\var{n}}}{\fundef{\app{\var{f}}{\var{n}}}{\cas{\brackets{\app{\app{\var{gt}}{\var{n}}}{\brackets{100}}}}{\var{True}}{\app{\app{\var{sub}}{\var{n}}}{10}}{\var{False}}{\app{\var{f}}{\brackets{\app{\var{f}}{\brackets{\app{\app{\var{plus}}{\var{n}}}{\brackets{11}}}}}}}} \\
\fundef{\app{\app{\var{gt}}{\var{x}}}{\var{y}}}{\cas{\var{x}}{\var{Zero}}{\var{False}}{\app{\var{Succ}}{\var{x}}}{\cas{\var{y}}{\var{Zero}}{\var{True}}{\app{\var{Succ}}{\var{y}}}{\app{\app{\var{gt}}{\var{x}}}{\var{y}}}}} \\
\fundef{\app{\app{\var{sub}}{\var{x}}}{\var{y}}}{\cas{\var{y}}{\var{Zero}}{\var{x}}{\app{\var{Succ}}{\var{y}}}{\cas{\var{x}}{\var{Zero}}{\var{Zero}}{\app{\var{Succ}}{\var{x}}}{\app{\app{\var{sub}}{\var{x}}}{\var{y}}}}} \\
\fundef{\app{\app{\var{plus}}{\var{x}}}{\var{y}}}{\cas{\var{x}}{\var{Zero}}{\var{y}}{\app{\var{Succ}}{\var{x}}}{\app{\var{Succ}}{\brackets{\app{\app{\var{plus}}{\var{x}}}{\var{y}}}}}} \\
}}$
\end{center}
\caption{McCarthy's 91 Function}
\label{mccarthy}
\end{figure}
Although the result of transforming this program using distillation (and the corresponding LTS) are too large to show here, we are also able to prove the termination of this program.
\end{example}

\section{Conclusion and Related Work}

In this paper, we have described a new approach to the termination analysis of functional programs that is applied to the output of the {\em distillation} 
program transformation \cite{HAMILTON07A,HAMILTON12}. Distillation converts programs into a simplified form called {\em distilled form},
and to prove that programs in this form terminate, we convert them into a corresponding {\em labelled transition system} and then show 
that all possible infinite {\em traces} through the labelled transition system would result in an infinite descent of well-founded data values. 

We argue that our termination analysis is simple and straightforward. We do not need to treat nested function calls, mutual recursion 
or permuted arguments as special cases, we do not need to search for appropriate ranking functions and we do not need to define a 
size ordering on values. Most recent approaches to proving termination have involved searching for a set of possible ranking functions 
and using a disjunctive termination argument \cite{BERDINE06,COOK06,MANOLIOS06,COOK11}. We avoid the need for such involved analysis here.

The most closely related work to that described here is that described in \cite{BROTHERSTON08} which makes use of {\em cyclic proof} techniques. 
In \cite{BROTHERSTON08}, a {\em cyclic pre-proof} form is defined that is a finite derivation tree in which every leaf that is not the conclusion of 
an axiom is closed by a backlink to a syntactically identical interior node. A global soundness condition is defined on pre-proofs so they can be verified 
as genuine {\em cyclic proofs}. This involves proving that every trace through the pre-proof is infinitely progressing and that there must therefore be 
infinite descent of the data values, in much the same way as is done in the work described here. The structure of a pre-proof is similar to the form of the 
labelled transition systems generated from programs that are in distilled form. However, a pre-proof does not contain any instances of the cut rule (which
correspond to {\bf let}s in distilled form) and therefore cannot have any accumulating parameters, so this technique is not applicable to as wide a range 
of programs as the technique described here. Also, because our programs have first been transformed into the form required to facilitate our proof, 
we are able to prove termination for an even wider class of programs.

Another closely related work to that described in this paper is the work on the {\em size-change principle} for termination \cite{LEE01}. 
In \cite{LEE01}, {\em size-change graphs} are created that indicate definite information about the change of size of parameters in function calls.
These graphs indicate whether a parameter is either decreasing or non-increasing. To prove termination of a program, it is then
necessary to show that every possible thread within a program is {\em infinitely descending}, meaning that it contains infinitely
many occurrences of a decreasing parameter. This is similar to the approach taken in this paper, where we also try to show that
there are infinitely many occurrences of a decreasing parameter. Both techniques can handle nested function calls. In \cite{LEE01},
these are handled directly and in this work they are transformed to remove this nesting prior to analysis. However, there are also a 
number of differences between these two techniques. Firstly, in \cite{LEE01}, if a parameter can possibly increase at any point in a 
thread, then it is not possible to determine whether it is infinitely descending. In this work, we can ignore this possibility as any parameter 
that decreases in size cannot previously have increased. However, in \cite{LEE01}, a more precise measure of parameter size is employed 
based on their semantic value with a well-founded partial ordering. In this work, a more conservative measure of the syntactic size of 
parameters is used, where a parameter that is assigned to any expression other than a variable is considered to increase in size. 
In our running example, in the call {\em gcd (sub x y) y}, it is difficult to determine that the first parameter is semantically decreasing as
we need to show that the number being subtracted is non-zero, so we cannot prove that it terminates using the size-change principle. 
Using our technique, even though we initially assume that this parameter is increasing as the syntactic size is increasing, 
the program is transformed by distillation to reveal that it is in fact decreasing, so we are able to prove termination. Also, our technique
is directly applicable to higher-order languages, while the size-change principle originally described in \cite{LEE01} is not.
An extension of the size-change principle to higher-order languages is described in \cite{SERENI05}, but it is far from straightforward.
We have not been able to find any examples of programs that can be found to terminate using this size-change principle
and cannot be found to terminate using the technique described here, but we have found many examples of the opposite being true. 
For example, none of the example programs in this paper can be shown to terminate using the size-change principle.

\bibliographystyle{eptcs}

\bibliography{mybib}

\end{document}